\documentclass[11pt, letter]{article}
\usepackage{cite,amsmath,amsfonts,subfigure,color,graphicx,url, amsthm}
\newtheorem{theorem}{Theorem}
\newtheorem{lemma}[theorem]{Lemma}
\newtheorem{corollary}[theorem]{Corollary}
\newtheorem{fact}[theorem]{Fact}
\newtheorem{observation}[theorem]{Observation}

\advance \topmargin by -\headheight
\advance \topmargin by -\headsep
\evensidemargin \oddsidemargin
\let\epsilon=\varepsilon
\newcommand{\R}{\mathbb{R}}
\newcommand{\mIn}{\mathop{\rm m\i n}}
\newcommand{\median}{\mathop{\rm median}}
\newcommand{\mean}{\mathop{\rm mean}}
\newbox\Dbox
\setbox\Dbox=\hbox{\scriptsize d}
\ht\Dbox=0pt
\newcommand{\meD}{{\mathop{\rm me\copy\Dbox}}}
\newcommand{\convolve}{{\mathop{*}\nolimits}}

\newcommand{\generalconvolve}{{\mathop{*}\limits_\oplus^\odot}}
\newcommand{\minconvolve}{{\mathop{*}\limits_{\mIn}^-}}
\newcommand{\maxconvolve}{{\mathop{*}\limits_{\max}^-}}
\newcommand{\medianconvolve}{{\mathop{*}\limits_{\meD}^-}}
\newcommand{\minmultiply}{{\mathop{\cdot}\limits_{\mIn}^-}}

\newcounter{last}

\makeatletter
\let\@fnsymbol=\@arabic
\makeatother

\begin{document}

\title{Necklaces, Convolutions, and $X+Y$}
\author{%
\begin{tabular}{c@{\quad\qquad}c@{\quad\qquad}c}
  David Bremner%
    \thanks{Faculty of Computer Science, University of New Brunswick,
      Fredericton, New Brunswick, Canada, \protect\url{bremner@unb.ca}.
      Supported by NSERC.}
&
  Timothy M. Chan%
    \thanks{School of Computer Science, University of Waterloo,
      Waterloo, Ontario, Canada, \protect\url{tmchan@uwaterloo.ca}.
      Supported by NSERC.}
&
  Erik D. Demaine%
    \thanks{Computer Science and Artificial Intelligence Laboratory,
      Massachusetts Institute of Technology, Cambridge, MA, USA,
      \protect\url{edemaine@mit.edu}.
      Supported in part by NSF grants CCF-0430849 and OISE-0334653 and
      by an Alfred P. Sloan Fellowship.}
\medskip\\
  Jeff Erickson%
    \thanks{Computer Science Department, University of Illinois,
      Urbana-Champaign, IL, USA, \protect\url{jeffe@cs.uiuc.edu}.}
&
  Ferran Hurtado%
    \thanks{Departament de Matem\`atica Aplicada II,
      Universitat Polit\`ecnica de Catalunya, Barcelona, Spain,
      \protect\url{Ferran.Hurtado@upc.edu}.
      Supported in part by  projects MICINN MTM2009-07242, Gen. Cat. DGR
      2009SGR1040, and ESF EUROCORES programme EuroGIGA,
      CRP ComPoSe: MICINN Project EUI-EURC-2011-4306, for Spain.}
&
  John Iacono%
    \thanks{Department of Computer and Information Science,
      Polytechnic University, Brooklyn, NY, USA,
      \protect\url{http://john.poly.edu}.
      Supported in part by NSF grants CCF-0430849 and OISE-0334653 and
      by an Alfred P. Sloan Fellowship.}
\medskip\\
  Stefan Langerman%
    \thanks{Directeur de Recherches du FRS--FNRS, D\'epartment d'Informatique,
      Universit\'e Libre de Bruxelles, Brussels, Belgium,
      \protect\url{stefan.langerman@ulb.ac.be}.}
&
 Mihai P\v{a}tra\c{s}cu%
    \thanks{Chercheur qualifi\'e du FNRS, D\'epartment d'Informatique,
      Universit\'e Libre de Bruxelles, Brussels, Belgium,
      \protect\url{stefan.langerman@ulb.ac.be}.}
&
 Perouz Taslakian%
    \thanks{College of Science and Engineering, American University of Armenia, Yerevan, Armenia,
      \protect\url{ptaslakian@aua.am}.}
\end{tabular}
}

\date{}

\maketitle

\vspace*{0.8cm} 
\begin{quote} 
\centering 
In memory of our colleague Mihai P\v{a}tra\c{s}cu. 
\end{quote}
\vspace*{0.8cm}

\begin{abstract}
  We give subquadratic algorithms that, given two necklaces each with $n$
  beads at arbitrary positions, compute the optimal rotation of the
  necklaces to best align the beads.
  Here alignment is measured according to the $\ell_p$ norm of the
  vector of distances between pairs of beads from opposite necklaces
  in the best perfect matching.
  We show surprisingly different results for $p=1$, 
  $p$ even, and $p=\infty$.
  For 
  $p$ even, we reduce the problem to standard convolution,
  while for
  $p=\infty$ and $p=1$, we reduce the problem to $(\min,+)$ convolution
  and $(\median,+)$ convolution.  Then we solve the latter two convolution
  problems in subquadratic time, which are interesting results in their
  own right.
  These results shed some light on the classic sorting $X+Y$ problem,
  because the convolutions can be viewed as computing order statistics on the
  antidiagonals of the $X+Y$ matrix.
  All of our algorithms run in $o(n^2)$ time, whereas the obvious algorithms
  for these problems run in $\Theta(n^2)$ time.
\end{abstract}

\newpage

\section{Introduction}

How should we rotate two necklaces, each with $n$ beads at different locations,
to best align the beads?  More precisely, each necklace is represented
by a set of $n$ points on the unit-circumference circle,
and the goal is to find rotations of the necklaces,
and a perfect matching between the beads of the two necklaces,
that minimizes some norm of the circular distances between matched beads.
In particular, the $\ell_1$ norm minimizes the average absolute circular
distance between matched beads, the $\ell_2$ norm minimizes the average
squared circular distance between matched beads, and the $\ell_\infty$ norm
minimizes the maximum circular distance between matched beads.
The $\ell_1$ version of this necklace alignment problem was introduced by
Toussaint \cite{Toussaint-2004-JCDCG} in the context of comparing
rhythms in computational music theory, with possible applications
to rhythm phylogeny
\cite{Diaz-2004, Toussaint-2004-ISMIR}.

Toussaint \cite{Toussaint-2004-JCDCG} gave a simple $O(n^2)$-time algorithm
for $\ell_1$ necklace alignment, and highlighted as an interesting open
question whether the problem could be solved in $o(n^2)$ time.
In this paper, we solve this open problem by giving $o(n^2)$-time algorithms
for $\ell_1$, $\ell_2$, and $\ell_\infty$ necklace alignment,
in both the standard real RAM model of computation and the less realistic
nonuniform linear decision tree model of computation.
Our results for the case of the $\ell_1$ and $\ell_\infty$ distance measures in the real RAM model
also answer the questions posed by Clifford et al. in~\cite{clifford-04} 
(see the \emph{shift matching problem} in Problem 5 of the tech report).

\paragraph{Necklace alignment problem.}

\begin{figure}[b]
  \centering
  \includegraphics[scale=0.86]{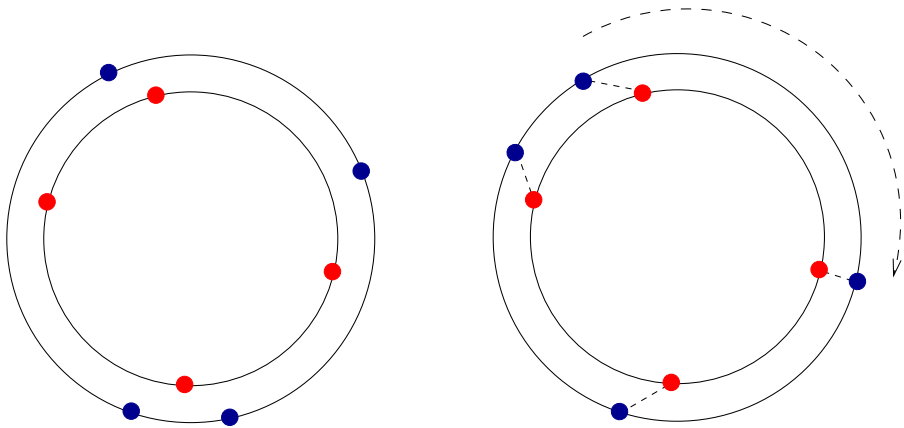}
  \caption{An example of necklace alignment: the input (left) and
    one possible output (right).}
  \label{necklace}
\end{figure}

More formally, in the \emph{necklace alignment problem}, the input is
a number $p$ representing the $\ell_p$ norm, and
two sorted vectors of $n$ real numbers,
$\vec x = \langle x_0, x_1, \dots, x_{n-1} \rangle$ and
$\vec y = \langle y_0, y_1, \dots, y_{n-1} \rangle$,
representing the two necklaces.
See Figure \ref{necklace}.
Canonically, we assume that each number $x_i$ and $y_i$ is in the range
$[0,1)$, representing a point on the unit-circumference circle
(parameterized clockwise from some fixed point).
The distance between two beads $x_i$ and $y_j$ is the minimum between the clockwise
and counterclockwise distances along the circumference of the unit-perimeter circular necklaces.
We define this distance as: 
$$d^\circ(x_i,y_j) = \min \{\left| x_i - y_j\right|, \displaystyle(1 - \left| x_i - y_j\right|\displaystyle)\}.$$ 

The optimization problem involves two parameters.
The first parameter, the \emph{offset} $c \in [0, 1)$, is the
clockwise rotation angle of the first necklace relative to the second necklace.
The second parameter, the \emph{shift} $s \in \{0, 1, \dots, n\}$,
defines the perfect matching between beads: bead $i$ of the first necklace
matches with bead $(i + s) \bmod n$ of the second necklace.
(Here we use the property that an optimal perfect matching between the beads
does not cross itself.)

The goal of the $\ell_p$ necklace alignment problem is to find
the offset $c \in [0, 1)$ and the shift $s \in \{0, 1, \dots, n\}$
that minimize
\begin{equation} 
\label{circ}
\sum_{i=0}^{n-1} \left(d^\circ((x_i+c) \bmod 1,y_{(i+s)\bmod n})\right)^p
\end{equation}
or, in the case $p=\infty$, that minimize
$$ \max_{i=0}^{n-1} \{d^\circ((x_i+c) \bmod 1,y_{(i+s)\bmod n})\}.$$ 

The $\ell_1$, $\ell_2$, and
$\ell_\infty$ necklace alignment problems all have trivial $O(n^2)$ solutions, 
although this might not be obvious from the definition.
In each case, as we show, the optimal offset $c$ can be computed
in linear time for a given shift value~$s$ (sometimes even independent of~$s$).
The optimization problem is thus effectively over just
$s \in \{0, 1, \dots, n\}$,
and the objective costs $O(n)$ time to compute for each~$s$,
giving an $O(n^2)$-time algorithm.

\paragraph{}
Although necklaces are studied throughout mathematics,
mainly in combinatorial settings, we are not aware of
any work on the necklace alignment problem before
Toussaint \cite{Toussaint-2004-JCDCG}.
He introduced $\ell_1$ necklace alignment, calling it the
\emph{cyclic swap-distance} or \emph{necklace swap-distance} problem,
with a restriction that the beads lie at integer coordinates.
Ardila et al.~\cite{ardila-2008} give a $O(k^2)$-time algorithm for computing the necklace swap-distance
between two binary strings, with $k$ being the number of 1-bits (beads at integer coordinates).
Colannino et
al.~\cite{Colannino-Damian-Hurtado-Iacono-Meijer-Ramaswami-Toussaint-2006}
consider some different distance measures between two sets of points
on the real line in which the matching does not have to match every point.
They do not, however, consider alignment under such distance measures.

Aloupis et al.~\cite{aloupis-2004}, consider the problem of computing the similarity
of two melodies represented as closed orthogonal chains on a cylinder.
Their goal is to find the proper (rigid) translation of one of the chains in the vertical (representing pitch)
and tangential (representing time) direction so that the area between the chains is minimized. 
The authors present an $O(m n \lg(n+m))$ algorithm that solves the problem.
When the melodic chains each have a note at every time unit, the  melodic similarity problem
is equivalent to the necklace alignment problem, and as our results are subquadratic,
we improve on the results of Aloupis et al.~\cite{aloupis-2004}
for this special case.

\paragraph{Convolution.}
Our approach in solving the necklace alignment problem 
is based on reducing it 
to another important problem, convolution, for which we also obtain improved algorithms.
The $(+,\cdot)$ convolution
of two vectors $\vec x = \langle x_0, x_1, \allowbreak\dots, x_{n-1} \rangle$ and
$\vec y = \langle y_0, y_1, \dots, y_{n-1} \rangle$, is the vector
$\vec x \convolve \vec y = \langle z_0, z_1, \dots, z_{n-1} \rangle$
where $z_k = \sum_{i=0}^k x_i \cdot y_{k-i}$.
One can generalize convolution to any $(\oplus,\odot)$ operators.
Algorithmically, a convolution with specified addition and
multiplication operators (here denoted $\vec x \generalconvolve \vec y$)
can be easily computed in $O(n^2)$ time.
However, 
the $(+,\cdot)$ convolution can be computed in $O(n \lg n)$
time using the Fast Fourier Transform
\cite{Cooley-Tukey-1965, Gauss-1866-III, Heideman-Johnson-Burrus-1985},
because the Fourier transform converts convolution into elementwise
multiplication.
Indeed, fast $(+,\cdot)$ convolution was one of the early breakthroughs
in algorithms, with applications to
polynomial and integer multiplication \cite{Bernstein-survey},
batch polynomial evaluation
\cite[Problem~30-5]{Cormen-Leiserson-Rivest-Stein-2001},
3SUM \cite{Erickson-1999-satisfiability, Baran-Demaine-Patrascu-2008},
string matching
\cite{cc-07, Fischer-Paterson-1973, Indyk-1998-string, Kalai-2002-dontcare,
Cole-Hariharan-2002},
matrix multiplication \cite{Cohn-Kleinberg-Szegedy-Umans-2005},
and even juggling \cite{Cardinal-Kremer-Langerman-2006}.

In this paper we use three types of convolutions: 
$(\min,+)$ convolution, whose $k$th entry $z_k = \min_{i=0}^k \left\{ x_i + y_{k-i} \right\}$; 
$(\median,+)$ convolution, whose $k$th entry $z_k = \median_{i=0}^k \left\{ x_i + y_{k-i} \right\}$;
and $(+,.)$ convolution, whose $k$th entry $z_k = \sum_{i=0}^k \left\{ x_i . y_{k-i} \right\}$.
As we show in
Theorems \ref{l_2}, \ref{l_infty reduction}, and \ref{l_1 reduction},
respectively,
$\ell_2$ necklace alignment reduces to standard $(+,\cdot)$ convolution,
$\ell_\infty$ necklace alignment reduces to $(\min,+)$ [and $(\max,+)$]
convolution, and
$\ell_1$ necklace alignment reduces to $(\median,+)$ convolution.
The $(\min,+)$ convolution problem has appeared frequently in the literature,
already appearing in Bellman's
early work on dynamic programming in the early 1960s
\cite{Bellman-Karush-1962, Felzenszwalb-Huttenlocher-2004,
Maragos-2000, Moreau-1970, Rockafellar-1970, Stroemberg-1996}.
Its name varies among ``minimum convolution'', ``min-sum convolution'',
``inf-convolution'', ``infimal convolution'', and ``epigraphical sum''.%
\footnote{``Tropical convolution'' would also make sense,
  by direct analogy with tropical geometry, but we have never
  seen this terminology used in print.}
To date, however, no worst-case $o(n^2)$-time algorithms for this convolution,
or the more complex $(\median,+)$ convolution, has been obtained
(it should be noted here that the quadratic worst-case running time 
for $(\median,+)$ convolution follows from linear-time median finding~\cite{Blum-1973, Schonhage-1976}).
In this paper, we develop worst-case $o(n^2)$-time algorithms for
$(\min,+)$ and $(\median,+)$ convolution,
in the real RAM and the nonuniform linear decision tree
models of computation.

The only subquadratic results for $(\min,+)$ convolution
concern two special cases.
First, the $(\min,+)$ convolution of two convex sequences or functions
can be trivially computed in $O(n)$ time by a simple merge,
which is the same as computing the Minkowski sum of two convex polygons
\cite{Rockafellar-1970}.
This special case is already used in image processing and computer vision
\cite{Felzenszwalb-Huttenlocher-2004, Maragos-2000}.
Second, Bussieck et al.~\cite{Bussieck-Hassler-Woeginger-Zimmermann-1994}
proved that the $(\min,+)$ convolution of two \emph{randomly permuted}
sequences can be computed in $O(n \lg n)$ expected time.
Our results are the first to improve the worst-case running time
for $(\min,+)$ convolution.

\paragraph{Connections to $X+Y$.}

The necklace alignment problems, and their corresponding convolution problems,
are also intrinsically connected to problems on $X+Y$ matrices.
Given two lists of $n$ numbers,
$X = \langle x_0, x_1, \dots, x_{n-1} \rangle$ and
$Y = \langle y_0, y_1, \dots, y_{n-1} \rangle$,
$X+Y$ is the matrix of all pairwise sums, whose $(i,j)$th entry is
$x_i + y_j$.
A classic unsolved problem \cite{TOPP-41}
is whether the entries of $X+Y$ can be sorted
in $o(n^2 \lg n)$ time.
Fredman \cite{Fredman-1976-XpY} showed that $O(n^2)$ comparisons suffice
in the nonuniform linear decision tree model, but it remains open whether
this can be converted into an $O(n^2)$-time algorithm in the real RAM model.
Steiger and Streinu \cite{Steiger-Streinu-1995} gave a simple algorithm
that takes $O(n^2 \lg n)$ time while using only $O(n^2)$ comparisons.

The $(\min,+)$ convolution is equivalent to finding the minimum
element in each antidiagonal of the $X+Y$ matrix, and similarly
the $(\max,+)$ convolution finds the maximum element in each antidiagonal.
We show that $\ell_\infty$ necklace alignment is equivalent to finding the
antidiagonal of $X+Y$ with the smallest \emph{range}
(the maximum element minus the minimum element).
The $(\median,+)$ convolution is equivalent to finding the median element
in each antidiagonal of the $X+Y$ matrix.
We show that $\ell_1$ necklace alignment is equivalent to finding the
antidiagonal of $X+Y$ with the smallest \emph{median cost}
(the total distance between each element and the median of the elements).
Given the apparent difficulty in sorting $X+Y$,
it seems natural to believe that the minimum, maximum, and median elements
of every antidiagonal cannot be found, and that the corresponding objectives
cannot be minimized, any faster than $O(n^2)$ total time.
Figure~\ref{xplusy} shows a sample $X+Y$ matrix with the maximum element
in each antidiagonal marked, with no apparent structure.
Nonetheless, we show that $o(n^2)$ algorithms are possible.

\begin{figure}[ht]
  \centering
  \includegraphics[scale=0.6]{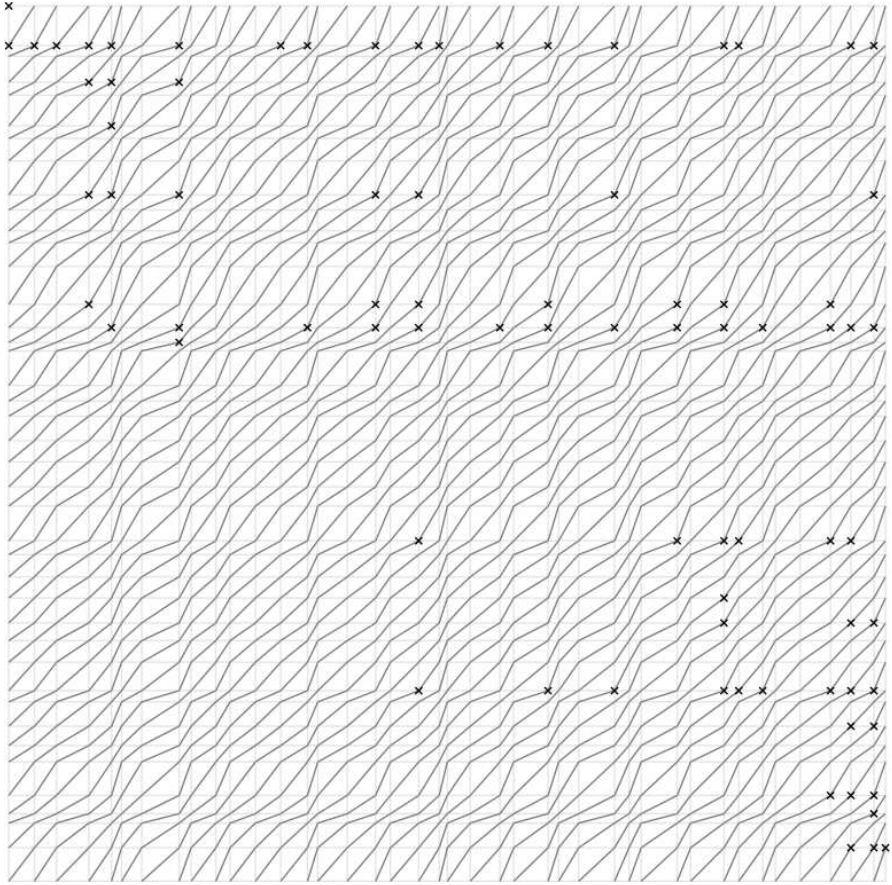}
  \caption{An $X+Y$ matrix.  Each polygonal line denotes an antidiagonal
           of the matrix, with a point at coordinates $(x,y)$ denoting
           the value $x+y$ for $x \in X$ and $y \in Y$.
           An~$\times$ denotes the maximum element in each antidiagonal.}
  \label{xplusy}
\end{figure}

\paragraph{Our results.}
In the standard real RAM model, we give subquadratic algorithms
for the $\ell_1$, $\ell_2$, and $\ell_\infty$ necklace alignment problems,
and for the $(\min,+)$ and $(\median,+)$ convolution problems.
We present:
\begin{enumerate}
\item an $O(n \lg n)$-time algorithm on the real RAM
      for $\ell_2$ necklace alignment (Section \ref{sec:l_2}).
\item an $O(n^2/\lg n)$-time algorithm on the real RAM
      for $\ell_\infty$ necklace alignment and $(\min,+)$ convolution
      (Section \ref{sec:l_infty}).
      This algorithm uses a technique of Chan originally developed for the
      all-pairs shortest paths problem~\cite{Chan-2005-apsp}.
      Despite the roughly logarithmic factor improvements for
      $\ell_1$ and~$\ell_\infty$, this result does not use word-level bit tricks
      of word-RAM fame.
\item a further improved $O(n^2 (\lg \lg n)^3 / \lg^2 n)$-time algorithm 
      for $\ell_\infty$ necklace alignment and $(\min,+)$ convolution
      (Section \ref{sec:l_infty}).
      We actually give a direct black-box reduction of $(\min,+)$ convolution
      to all-pairs shortest paths; the result then follows from the
      current best upper bound for all-pairs shortest paths~\cite{Chan-2010-apsp}.
      The all-pairs shortest paths works in the real RAM with respect to the inputs, i.e. it does not use bit tricks on the inputs. 
      The algorithm, however, requires bit tricks on other numbers, but works in a standard model
      that assumes $(\lg n)$-bit words.
\item an $O(n^2 (\lg \lg n)^2/\lg n)$-time algorithm on the real RAM
      for $\ell_1$ necklace alignment and $(\median,+)$ convolution
      (Section \ref{sec:l_1}).
      This algorithm uses an extension of the technique of Chan~\cite{Chan-2005-apsp}.
\setcounter{last}{\theenumi}
\end{enumerate}
In the nonuniform linear decision tree model, we give particularly fast
algorithms for the $\ell_1$ and $\ell_\infty$ necklace alignment problems,
using techniques of Fredman \cite{Fredman-1976-XpY, Fredman-1976-APSP}:
\begin{enumerate}
\setcounter{enumi}{\thelast}
\item $O(n \sqrt n)$-time algorithm in the nonuniform linear decision
      tree model for $\ell_\infty$ necklace alignment and
      $(\min,+)$ convolution (Section \ref{sec:l_infty}).
\item $O(n \sqrt {n \lg n})$-time algorithm in the nonuniform linear
      decision tree model for $\ell_1$ necklace alignment and
      $(\median,+)$ convolution (Section \ref{sec:l_1}).
\end{enumerate}
(Although we state our results here in terms of $(\min,+)$ and
$(\median,+)$ convolution, the results below use $-$ instead of $+$
because of the synergy with necklace alignment.)
We also mention connections to the venerable $X+Y$ and 3SUM problems
in Section~\ref{Conclusion}.

\section{Linear Versus Circular Alignment}
Before we proceed with proving our results, we first show that 
any optimal solution to the necklace alignment problem can be transformed into an optimal solution to 
the problem of linear alignment---aligning and matching beads that are on a line. 
We then can use the  simpler optimization function of the ``linear alignment problem'' to show our results.
Let $d^- (x_i,y_j)= |x_i - y_j|$ be the linear distance between two beads $x_i$ and $y_j$.
In the \emph{linear alignment problem} we 
are given two sorted vectors of real numbers 
$\vec x = \langle x_0, x_1, \dots, x_{n-1} \rangle$ and
$\vec y = \langle y_0, y_1, \dots, y_{m-1} \rangle$
with $m \geq n$,
and we want to find $s$ ($s<m-n$) and $c$ that minimize 
\begin{equation} 
\label{lin}
\sum_{i=0}^{n-1} \left(d^-(x_i+c,y_{i+s})\right)^p
\end{equation}
or, in the case $p=\infty$, that minimize
$$ \max_{i=0}^{n-1} \{d^-(x_i+c,y_{i+s})\}.$$ 

The main difference between \eqref{circ} and \eqref{lin} is that instead of taking
the minimum between the clockwise and counterclockwise distances between pairs of matched beads in \eqref{circ},
we are simply summing the forward distances between beads in \eqref{lin}. 
We will now show that 
whether the beads are on a line (repeating $\vec y$ infinitely many times on the line) or a circle, the optimal alignment of the beads $\vec x$ and $\vec y$ 
in these two cases are equal.

Let $M^{\circ}$ and $M^-$ be an alignment/matching of the beads of $\vec x$ and $\vec y$ along the unit circumference circle $C$ and 
the infinite line segment $L$ respectively.
An \emph{edge} $(x_i,y_j)$ of $M^{\circ}$ ($M^-$)
is the shortest segment that connects two matched beads $x_i$ and $y_j$ in $M^{\circ}$ ($M^-$);
thus, the length of $(x_i,y_j)$ is equal to $d^\circ(x_i,y_j)$ ($d^-(x_i,y_j)$).
We will show that the sum of the lengths of the edges of each of the optimal matchings $M^{\circ*}$ and $M^{-*}$ are equal.
Note that by the quadrangle inequality, we have that
the edges of both $M^{\circ*}$ and $M^{-*}$ are non-crossing.

\begin{observation}
\label{min edge length}
Consider any edge $(x_i,y_j)$ along the circular necklace. 
If this edge crosses point $0$, 
then the distance $$d^\circ(x_i,y_j) = \displaystyle(1-|x_i - y_j|\displaystyle);$$
otherwise, $$d^\circ(x_i,y_j) = |x_i - y_j|.$$
\end{observation}

Let 
$\vec{yy}$ be the doubling of the vector~$\vec y$ such that
\begin{eqnarray*}
\vec{yy} &=& \langle y_0, \dots, y_{m-1}, y_0, \dots, y_{m-1} \rangle
= \langle yy_0, \dots, yy_{m-1}, yy_m, \dots, yy_{2m-1} \rangle.\\
\end{eqnarray*}

\begin{theorem}
If $M^{\circ *}$ is the optimal matching of two given vectors $\vec x$ and $\vec y$ (both of length $n$) along the unit-circumference circle $C$
and $M^{-*}$ is the optimal matching of $\vec{x}$ and $\vec{yy}$ along line~$L$, then $|M^{\circ*}| = |M^{-*}|$.
\end{theorem}

\begin{proof}
First we show that the value of any optimal matching of a set of beads along $L$ is at least 
as large as the value of the optimal solution of the beads along $C$.
Given an optimal matching $M^{-*}=\{s,c\}$ beads along $L$, 
we will ``wrap'' the line $L$ around a unit circle $C$
by mapping each of the $n$ \emph{matched} 
beads $x_i$ ($yy_{i+s}$)  along~$L$ to $x_i^\circ$ ($y_{(i+s)\bmod n}^\circ$) along~$C$.
(Thus we have exactly $n$ pairs of beads along~$C$.) 
Now, for every $i=0,1,\dots,n-1$, the length of an edge of $M^{-*}$ is equal to
\begin{eqnarray*}
&&d^-(x_i+c,yy_{i+s}) \\
&=&d^-(x_i+c,y_{(i+s) \bmod n}) \\ 
&=& |x_i+c - y_{(i+s)\bmod n}|\\ 
&\geq& \min \{\left| (x_i + c)\bmod 1 - y_{(i+s)\bmod n}\right| , \displaystyle(1 - \left| (x_i + c)\bmod 1 - y_{(i+s)\bmod n}\right|\displaystyle)\}\\ 
&=& d^\circ((x_i^\circ + c)\bmod 1, y_{(i+s)\bmod n}^\circ).
\end{eqnarray*}
Thus, as every edge length of the matching $M^{-*}$ is at least as large as its corresponding edge along the circle $C$, we have $|M^{-*}| \geq |M^{\circ*}|$.

Next we show that the value of any optimal matching of a set of beads along $C$ is at least 
as large as the value of the optimal solution of the beads along $L$.
Suppose we have an optimal matching $M^{\circ*}=(s,c)$. 
We map every point $x_i+c$ and $y_i$ to the infinite line segment so that
the edges of $M^{\circ*}$ are  preserved in $M^{-}$.
Thus, for all $i = 0,1,\dots,n-1$ and $k \in \mathbb{Z}$, we map
\begin{eqnarray*}
(x_i + c)\bmod 1 &\mapsto& x^-_{i} = x_i + c, \\  
y_i &\mapsto& y^-_{i+kn} = y_i + k.\\
\end{eqnarray*}

With this transformation, in any valid matching $M^{-}$, the matched beads span at most two consecutive intervals $[k,k+1)$ and $[k+1,k+2)$
for any $k \in \mathbb{Z}$. In particular, the beads $x_i + c$ span the intervals $[0,1)$ and $[1,2)$.

Now we construct $M^{-}$ given $M^{\circ*}$ by matching every $x^-_{i}$ to $y^-_{i+s+kn}$ such that,
whenever $x_i+c < 1$ (see Figure~\ref{unrolling}),
\begin{itemize}
\item $k=-1$ \\ if $((x_i+c)\bmod 1, y_{(i+s)\bmod n})$ crosses point $0$ and $(x_i+c)\bmod 1 < y_{(i+s)\bmod n}$;
\item $k=1$ \\ if $((x_i+c)\bmod 1, y_{(i+s)\bmod n})$ crosses point $0$ and $(x_i+c)\bmod 1 > y_{(i+s)\bmod n}$;
\item $k=0$ \\ if $((x_i+c)\bmod 1, y_{(i+s)\bmod n})$ does not cross point $0$;
\end{itemize}
and, whenever $x_i+c \geq 1$, we increment $k$ by $1$ in each of the cases.
Thus, when $x_i+c \geq 1$,
\begin{itemize}
\item $k=0$ \\ if $((x_i+c)\bmod 1, y_{(i+s)\bmod n})$ crosses point $0$ and $(x_i+c)\bmod 1 < y_{(i+s)\bmod n}$;
\item $k=2$ \\ if $((x_i+c)\bmod 1, y_{(i+s)\bmod n})$ crosses point $0$ and $(x_i+c)\bmod 1 > y_{(i+s)\bmod n}$;
\item $k=1$ \\ if $((x_i+c)\bmod 1, y_{(i+s)\bmod n})$ does not cross point $0$.
\end{itemize}

\begin{figure}[t]
\centering
\subfigure[A matching $M^{\circ}$ of beads along a circular necklace with three types of edges: 
   those that do not cross point $0$, 
   and those that cross point $0$ 
   with $(x_i+c) \bmod 1$ either greater or less than~$y_{i+s}$.]{
   \includegraphics[scale=0.7]{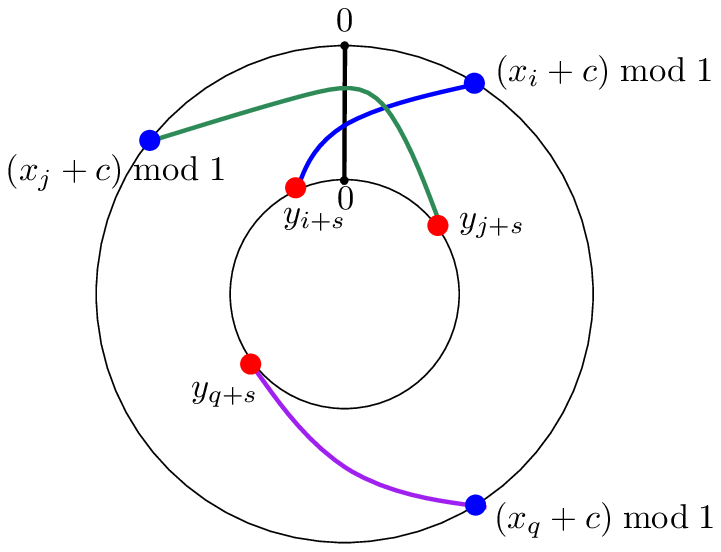}
   \label{fig:circle}
 }\hskip0.35cm
 \subfigure[Unrolling a necklace to a line $L$ by preserving the edges of $M^{\circ}$ in the  matching $M^{-}$.  
   In this example, $0 \leq x_0+c < 1$.]{
   \includegraphics[scale=0.8]{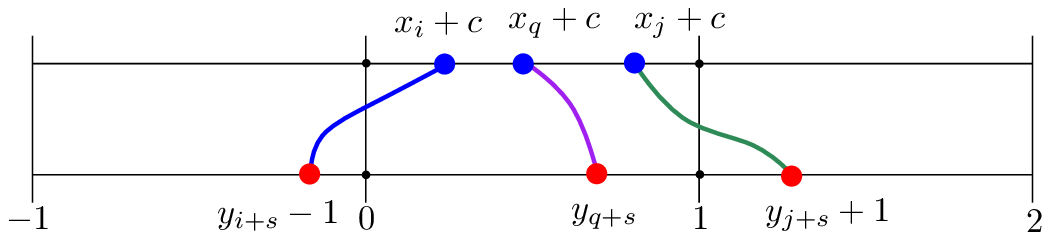}
   \label{fig:line}
 }
\caption{Unrolling a circular necklace to a line.}
\label{unrolling}
\end{figure}

Here, the variable $k$ basically decides the interval $[-1,0)$, $[0,1)$, or $[1,2)$ in 
which the bead $y_{(i+s)\bmod n}$ is located, based on the type of the edge $((x_i+c)\bmod 1, y_{(i+s)\bmod n})$.
Observe that, if an edge in $M^{\circ*}$ crosses point $0$, 
then its corresponding edge in $M^{-}$ crosses $(r,r)$ for some $r \in \{0,1,2\}$.

Now, the sum of the distances of the matched beads of $M^{-}$ is equal to
\begin{align*}
&d^-(x_0+c, y_{(0+s) \bmod n}+k) + d^-(x_1+c, y_{(1+s) \bmod n}+k) + \dots + \\
&d^-(x_{n-s-1}+c, y_{n-1}+k)+ d^-(x_{n-s}+c, y_{0}+k+1) +\dots + \\
&d^-(x_{n-s-1}+c, y_{n-1}+k+1).
\end{align*}

We claim that the value of $M^{\circ*}$ is equal to at least the value of this matching $M^{-}$. 
We show this claim by comparing the length of each edge $((x_i+c)\bmod 1, y_{i+s})$ of $M^{\circ*}$ with its corresponding edge in $M^{-}$.

\paragraph{Edges that do not cross point \boldmath{$0$}:}
If an edge of $M^{\circ*}$ does not cross point $0$, then the corresponding edge in $M^{-}$ 
does not cross any of the edges $(0,0)$, $(1,1)$ or $(2,2)$; 
hence both endpoints (beads) of the given matching edge are within the same interval $[r, r+1)$ for some $r\in\{0,1\}$
(the purple edge $(x_q+c,y_{q+s})$ in Figure~\ref{unrolling}).
This means that, when $x_i+c = (x_i+c)\bmod 1$, we have $k=0$ and 
$$
\begin{array}{ll}
d^-(x_i+c, y_{(i+s)\bmod n}+k) &= |(x_i+c) - (y_{(i+s)\bmod n}+k)| \\
&= |(x_i+c)\bmod 1 - (y_{(i+s)\bmod n})| \\
&= |x_i^- - y^-_{(i+s)\bmod n}| \\
&= d^\circ(x_i^-, y_{(i+s)\bmod n}^-)  ${\rm \indent(by~Observation~\ref{min edge length}).}$\\
\end{array}
$$

We can similarly show that, when $x_i+c = (x_i+c)\bmod 1 + 1$, we have $k=1$ and the
edges  of $M^{\circ*}$ that do not cross point $0$ have the same length as their corresponding edge in $M^{-}$.

\paragraph{Edges that cross point \boldmath{$0$}:}
If an edge of $M^{\circ*}$ crosses point $0$, then the corresponding edge 
in $M^{-}$ crosses edge $(r,r)$ and hence 
the two endpoints (beads) of the given edge of $M^{-}$ must be in different and consecutive intervals: 
$[r-1, r)$ and $[r, r+1)$ for some $r\in \{0,1,2\}$
(the green and blue edges $(x_i+c,y_i+s)$,$(x_j+c,y_j+s)$ in Figure~\ref{unrolling}).
Then, assuming $x_i + c = (x_i+c)\bmod 1$, we have

\begin{itemize}
\item
When $(x_i+c)\bmod 1 < y_{(i+s)\bmod n}$, $k = -1$ and 
$$
\begin{array}{ll}
d^-(x_i+c, y_{(i+s)\bmod n}+k) &= |(x_i+c)  - (y_{(i+s)\bmod n}-1)| \\
&= |(x_i+c)\bmod 1 - y_{(i+s)\bmod n} + 1|\\
&= \displaystyle(1 - |((x_i+c)\bmod 1) - y_{(i+s)\bmod n}|\displaystyle)\\
&= \displaystyle(1 - |x_i^- - y^-_{((i+s)\bmod n)-n}|\displaystyle)\\
&= d^\circ(x_i^-, y^-_{((i+s)\bmod n)-n}) ${\rm \indent(by~Observation~\ref{min edge length}).}$\\
\end{array}
$$

\item
When $(x_i+c)\bmod 1 > y_{(i+s)\bmod n}$, $k = 1$ and 
$$
\begin{array}{ll}
d^-(x_i+c, y_{(i+s)\bmod n}+k) &= |(x_i+c)  - (y_{(i+s)\bmod n}+1)| \\
&= |(x_i+c)\bmod 1 - y_{(i+s)\bmod n} - 1|\\
&= \displaystyle(1 - |((x_i+c)\bmod 1) - y_{(i+s)\bmod n}|\displaystyle)\\
&= \displaystyle(1 - |x_i^- - y^-_{((i+s)\bmod n)+ n}|\displaystyle)\\
&= d^\circ(x^-_i+c, y^-_{((i+s)\bmod n)+n})  ${\rm \indent(by~Observation~\ref{min edge length}).}$\\
\end{array}
$$
We can similarly show that, when $x_i+c = (x_i+c)\bmod 1 + 1$,
the edges  of $M^{\circ*}$ that cross point $0$ have the same length as their corresponding edge in $M^{-}$.

\end{itemize}

Therefore, the length of every edge of $M^{\circ*}$ along the circle is equal to the length of its corresponding edge in $M^{-}$.
Thus, the value of the matching $M^{\circ*}$ is at least as large as that of $M^{-*}$, completing the proof of the theorem.
\end{proof}

We now proceed to prove our results by using the objective function \eqref{lin}.

\section{$\ell_2$ Necklace Alignment and $(+,\cdot)$ Convolution}
\label{sec:l_2}

In this section, we first show how $\ell_2$ necklace alignment reduces
to standard convolution, leading to an $O(n \lg n)$-time algorithm that uses the Fast Fourier Transform.
We then show how this result generalizes to $\ell_p$ for any even $p$.
It should be noted here that the $\ell_2$ necklace alignment problem was solved independently by Clifford et al.~\cite{clifford-04} (see Problem 5)
using Fast Fourier Transforms. Our proof uses essentially the same technique of expanding the squared term and then optimizing terms separately,
but goes through the steps in more detail;
we include our proof for completeness. 
More results that use the FFT to solve different flavors of matching problems may also be found in~\cite{clifford-05} and \cite{clifford-07}.

\begin{theorem} \label{l_2}
  The $\ell_2$ necklace alignment problem can be solved in $O(n \lg n)$ time
  on a real RAM.
\end{theorem}

\begin{proof}
  The objective \eqref{lin} expands algebraically to
  \begin{eqnarray*}
    &&\sum_{i=0}^{n-1} \left( x_i - y_{(i+s) \bmod n} + c \right)^2\\
    &=& \sum_{i=0}^{n-1} \left( x_i^2 + y_{(i+s) \bmod n}^2 +
                                2 c x_i - 2c y_{(i+s) \bmod n} + c^2 \right)
      - 2 \sum_{i=0}^{n-1} x_i y_{(i+s) \bmod n} \\
    &=& \sum_{i=0}^{n-1} \left( x_i^2 + y_i^2 +
                                2 c x_i - 2c y_i + c^2 \right)
      - 2 \sum_{i=0}^{n-1} x_i y_{(i+s) \bmod n} \\
    &=& \left[
          \sum_{i=0}^{n-1} \left( x_i^2 + y_i^2 \right) +
          2 c \sum_{i=0}^{n-1} \left( x_i - y_i \right) +
          n c^2
        \right]
      - 2 \sum_{i=0}^{n-1} x_i y_{(i+s) \bmod n}. \\
  \end{eqnarray*}
  The first term depends solely on the inputs and the variable~$c$,
  while the second term depends solely on the inputs and the variable~$s$.
  Thus the two terms can be optimized separately.
  The first term can be optimized in $O(n)$ time by solving for when the
  derivative, which is linear in~$c$, is zero.
  The second term can be computed, for each $s \in \{0, 1, \dots, n-1\}$,
  in $O(n \lg n)$ time using $(+,\cdot)$ convolution
  (and therefore optimized in the same time).
  Specifically, define the vectors
  \begin{eqnarray*}
    \vec x' &=& \langle x_0, x_1, \dots, x_{n-1};
                        \underbrace{0, 0, \dots, 0}_n \rangle, \\
    \vec y' &=& \langle y_{n-1}, y_{n-2}, \dots, y_0;
                        y_{n-1}, y_{n-2}, \dots, y_0 \rangle.
  \end{eqnarray*}
  Then, for $s' \in \{0, 1, \dots, n-1\}$, the $(n+s')$th entry of the
  convolution $\vec x' \convolve \vec y'$ is
  $$ \sum_{i=0}^{n+s'} x'_i y'_{n+s'-i}
   = \sum_{i=0}^{n-1} x_i y_{(i-s'-1) \bmod n},
  $$
  which is the desired entry if we let $s' = n-1-s$.
  We can compute the entire convolution in $O(n \lg n)$ time
  using the Fast Fourier Transform.
\end{proof}

The above result can be generalized to $\ell_p$ for any fixed \emph{even} integer $p$.
When $p \geq 4$, expanding the objective and rearranging the terms results in
\begin{align*}
\sum_{i=0}^{n-1} (x_i-y_{(i+s)\bmod n} + c)^p &= \sum_{i=0}^{n-1}\sum_{j=0}^{p}{p \choose j}(x_i-y_{(i+s)\bmod n})^{p-j}c^j\\
&=\sum_{j=0}^{p}\left({p \choose j}\sum_{i=0}^{n-1}(x_i-y_{(i+s)\bmod n})^{p-j}\right)c^j,\\
\end{align*}
which is a degree-$p$ polynomial in $c$, all of whose coefficients can be computed
for all values of $s$ by computing $O(p^2)$ convolutions.

\begin{theorem} \label{l_p}
  The $\ell_p$ necklace alignment problem with $p$ even can be solved in $O(p^2n \lg n)$ time
  on a real RAM.
\end{theorem}

\section{$\ell_\infty$ Necklace Alignment and $(\min,+)$ Convolution}
\label{sec:l_infty}

\subsection{Reducing $\ell_\infty$ Necklace Alignment to $(\min,+)$ Convolution}

First we show the relation between $\ell_\infty$ necklace alignment
and $(\min,+)$ convolution.  We need the following basic fact:

\begin{fact} \label{min max fact}
  For any vector $\vec z = \langle z_0, z_1, \dots, z_{n-1} \rangle$ and 
  $c = -\frac{1}{2} \left( \min_{i=0}^{n-1} z_i + \max_{i=0}^{n-1} z_i \right)$,
  the minimum value of $\max_{i=0}^{n-1} |z_i + c|$
  is $$\frac{1}{2} \left( \max_{i=0}^{n-1} z_i - \mIn_{i=0}^{n-1} z_i \right).$$
\end{fact}

Instead of using $(\min,+)$ convolution directly,
we use two equivalent forms, $(\min,-)$ and $(\max,-)$ convolution:

\begin{theorem} \label{l_infty reduction}
  The $\ell_\infty$ necklace alignment problem can be reduced in $O(n)$ time
  to one $(\min,-)$ convolution and one $(\max,-)$ convolution.
\end{theorem}

\begin{proof}
  For two necklaces $\vec x$ and $\vec y$, we apply the $(\min,-)$ convolution
  to the following vectors:
  \begin{eqnarray*}
    \vec x' &=& \langle x_0, x_1, \dots, x_{n-1};
                        \underbrace{\infty, \infty, \dots, \infty}_n \rangle, \\
    \vec y' &=& \langle y_{n-1}, y_{n-2}, \dots, y_0;
                        y_{n-1}, y_{n-2}, \dots, y_0 \rangle.
  \end{eqnarray*}
  Then, for $s' \in \{0, 1, \dots, n-1\}$, the $(n+s')$th entry of
  $\vec x' \minconvolve \vec y'$ is
  $$ \mIn_{i=0}^{n+s'} (x'_i - y'_{n+s'-i})
   = \mIn_{i=0}^{n-1} (x_i - y_{(i-s'-1) \bmod n}),
  $$
  which is $\min_{i=0}^{n-1} (x_i - y_{(i+s) \bmod n})$
  if we let $s' = n-1-s$.
  By symmetry, we can compute the $(\max,-)$ convolution
  $\vec x'' \maxconvolve \vec y'$, where $\vec x''$ has $-\infty$'s
  in place of $\infty$'s, and use it to compute
  $\max_{i=0}^{n-1} (x_i - y_{(i+s) \bmod n})$
  for each $s \in \{0, 1, \dots, n-1\}$.
  Applying Fact~\ref{min max fact}, we can therefore minimize
  $\max_{i=0}^{n-1} |x_i - y_{(i+s) \bmod n} + c|$ over~$c$,
  for each $s \in \{0, 1, \dots, n-1\}$.
  By brute force, we can minimize over $s$ as well using $O(n)$
  additional comparisons and time.
\end{proof}

\subsection{$(\min,-)$ Convolution in Nonuniform Linear Decision Tree}

For our nonuniform linear decision tree results, we use the main theorem
of Fredman's work on sorting $X+Y$:

\begin{theorem} {\rm \cite{Fredman-1976-XpY}} \label{Fredman lemma}
  For any fixed set $\Gamma$ of permutations of $N$ elements,
  there is a comparison tree of depth $O(N + \lg |\Gamma|)$ that sorts
  any sequence whose rank permutation belongs to $\Gamma$.
\end{theorem}

\begin{theorem} \label{min convolution nonuniform}
  The $(\min,-)$ convolution of two vectors of length $n$
  can be computed in $O(n \sqrt n)$ time
  in the nonuniform linear decision tree model.
\end{theorem}

\begin{proof}
  Let $\vec x$ and $\vec y$ denote the two vectors of length~$n$,
  and let $\vec x \minconvolve \vec y$ denote their $(\min,-)$ convolution,
  whose $k$th entry is $\min_{i=0}^k \left( x_i - y_{k-i} \right)$.

  First we sort the set $D = \{x_i - x_j, y_i - y_j : |i-j| \leq d\}$
  of pairwise differences between nearby $x_i$'s and nearby $y_i$'s,
  where $d \leq n$ is a value to be determined later.
  This set $D$ has $N = O(n d)$ elements.
  The possible sorted orders of $D$ correspond to cells in the arrangement of
  hyperplanes in $\R^{2 n}$ induced by all $N \choose 2$ possible comparisons
  between elements in the set,
  and this hyperplane arrangement has $O(N^{4 n})$ cells.
  By Theorem~\ref{Fredman lemma}, there is a comparison tree sorting $D$
  of depth $O(N + n \lg N) = O(n d + n \lg n)$.

  The comparisons we make to sort $D$ enable us to compare $x_i - y_{k-i}$
  versus $x_j - y_{k-j}$ for free, provided $|i - j| \leq d$, because
  $x_i - y_{k-i} < x_j - y_{k-j}$ precisely if
  $x_i - x_j < y_{k-i} - y_{k-j}$.
  Thus, in particular, we can compute
  $$M_k(\lambda) = \min \Big\{x_i - y_{k-i} ~\Big|~
                 i = \lambda, \lambda+1, \dots, \min\{\lambda+d,n\}-1\Big\}$$
  for free (using the outcomes of the comparisons we have already made).

  We can rewrite the $k$th entry
  $\min_{i=0}^k ( x_i - y_{k-i} )$ of $\vec x \minconvolve \vec y$
  as $\min\{M_k(0), \allowbreak M_k(d), \allowbreak M_k(2 d), \dots, M_k(\lceil k/d \rceil d)\}$,
  and thus we can compute it in $O(k/d) = O(n/d)$ comparisons
  between differences.
  Therefore all $n$ entries can be computed in $O(n d + n \lg n + n^2/d)$
  total time.

  This asymptotic running time is minimized when $n d = \Theta(n^2/d)$,
  i.e., when $d^2 = \Theta(n)$.
  Substituting $d = \sqrt{n}$, we obtain a running time of
  $O(n \sqrt n)$ in the nonuniform linear decision tree model.
\end{proof}

Combining Theorems~\ref{l_infty reduction} and \ref{min convolution nonuniform},
we obtain the following result:

\begin{corollary} \label{l_infty nonuniform}
  The $\ell_\infty$ necklace alignment problem can be solved in
  $O(n \sqrt n)$ time in the nonuniform linear decision tree model.
\end{corollary}

\subsection{$(\min,-)$ Convolution in Real RAM via Geometric Dominance}

Our algorithm 
on the real RAM uses the following geometric lemma from
Chan's work on all-pairs shortest paths:

\begin{lemma} {\rm \cite[Lemma~2.1]{Chan-2005-apsp}} \label{chan}
  Given $n$ points $p_1, p_2, \dots, p_n$ in $d$ dimensions, each colored
  either red or blue, we can find the $P$ pairs $(p_i, p_j)$ for which
  $p_i$ is red, $p_j$ is blue, and $p_i$ dominates~$p_j$ (i.e., for all~$k$,
  the $k$th coordinate of $p_i$ is at least the $k$th coordinate of $p_j$), in
  $2^{O(d)} n^{1+\epsilon} + O(P)$ time for arbitrarily small $\epsilon > 0$.
\end{lemma}

\begin{theorem} \label{min convolution RAM}
  The $(\min,-)$ convolution of two vectors of length $n$
  can be computed in $O(n^2/\lg n)$ time on a real RAM.
\end{theorem}

\begin{proof}
  Let $\vec x$ and $\vec y$ denote the two vectors of length~$n$,
  and let $\vec x \maxconvolve \vec y$ denote their $(\max,-)$ convolution.
  (Symmetrically, we can compute the $(\min,-)$ convolution.)
  for each $i \in \{0, d, 2 d, \dots, \lfloor n/d \rfloor d\}$,
  and for each $j \in \{0, 1, \dots, n-1\}$,
  we define the $d$-dimensional points
  $$\arraycolsep=0.5\arraycolsep
  \begin{array}{rcllll}
    p_{\delta,i} &=& (x_{i+\delta} - x_i, & x_{i+\delta} - x_{i+1}, 
& \dots, & x_{i+\delta} - x_{i+d-1}), \\
    q_{\delta,j} &=& (y_{j-\delta} - y_i, & y_{j-\delta} - y_{i-1}, 
& \dots, &  y_{j-\delta} - y_{j-d-1}).
  \end{array}$$
  (To handle boundary cases, define $x_i=\infty$ and $y_j=-\infty$
  for indices $i,j$ outside $[0,n-1]$.)
  For each $\delta \in \{0,1,\dots, d-1\}$,
  we apply Lemma~\ref{chan} to the set of red
  points $\{p_{\delta,i} : i = 0, d, 2 d, \dots, \lfloor n/d \rfloor d\}$ and
  the set of blue points $\{q_{\delta,j} : j = 0, 1, \dots, n-1\}$,
  to obtain all dominating pairs $(p_{\delta,i},q_{\delta,j})$.

  Point $p_{\delta,i}$ dominates $q_{\delta,j}$ precisely if
  $x_{i+\delta} - x_{i+\delta'} \geq y_{j-\delta} - y_{j-\delta'}$
  for all $\delta' \in \{0, 1, \dots, d-1\}$
  (ignoring the indices outside $[0,n-1]$).
  By re-arranging terms, this condition is equivalent to
  $x_{i+\delta} - y_{j-\delta} \geq x_{i+\delta'} - y_{j-\delta'}$
  for all $\delta' \in \{0, 1, \dots, d-1\}$.
  If we substitute $j = k-i$, we obtain that 
  $(p_{\delta,i}, q_{\delta,k-i})$ is a dominating
  pair precisely if 
  $x_{i+\delta} - y_{k-i-\delta} = \max_{\delta'=1}^{d-1} (x_{i+\delta'} - y_{k-i-\delta'})$.
  Thus, the set of dominating pairs gives us the maximum 
  $M_k(i)=\max\{x_i - y_{k-i}, x_{i+1} - y_{k-i+1}, \dots,
           x_{\min\{i+d,n\}-1} - y_{\min\{k-i+d,n\}-1}\}$ 
  for each $i$ divisible by $d$ and for each~$k$.
  Also, there can be at most $O(n^2/d)$ such pairs for all $i,j,\delta$, 
  because there are $O(n/d)$ choices for $i$~and
  $O(n)$ choices for~$j$, and
  if $(p_{\delta,i},q_{\delta,j})$
  is a dominating pair, then $(p_{\delta',i},q_{\delta',j})$ cannot be a
  dominating pair for any $\delta'\neq \delta$.
  (Here we assume that the $\max$ is achieved uniquely, which can be arranged
   by standard perturbation techniques or by breaking ties consistently
   \cite{Chan-2005-apsp}.)
  Hence, the running time of the $d$ executions of Lemma~\ref{chan} is
  $d 2^{O(d)} n^{1+\epsilon} + O(n^2/d)$ time,
  which is $O(n^2/\lg n)$ if we choose $d = \alpha \lg n$
  for a sufficiently small constant $\alpha > 0$.
  We can rewrite the $k$th entry
  $\max_{i=0}^k ( x_i - y_{k-i} )$ of $\vec x \maxconvolve \vec y$ as
  $\max\{M_k(0), M_k(d), M_k(2 d), \allowbreak \dots, M_k(\lceil k/d \rceil d)\}
$,
  and thus we can compute it in $O(k/d) = O(n/d)$ time.
  Therefore all $n$ entries can be computed in $O(n^2/d) = O(n^2/\lg n)$ time
  on a real RAM.
\end{proof}

Combining Theorems~\ref{l_infty reduction} and \ref{min convolution RAM},
we obtain the following result:

\begin{corollary} \label{l_infty RAM}
  The $\ell_\infty$ necklace alignment problem can be solved
  in $O(n^2/\lg n)$ time on a real RAM.
\end{corollary}

Although we will present a slightly faster algorithm for $(\min,-)$ convolution 
in the next subsection, the approach described above will be useful
later when we discuss the $(\median,-)$ convolution problem.

\subsection{$(\min,-)$ Convolution 
via Matrix Multiplication}

Our next algorithm 
uses Chan's $O(n^3 (\lg \lg n)^3 / \lg^2 n)$
algorithm for computing the $(\min,+)$ matrix multiplication of two
$n \times n$ matrices \cite{Chan-2010-apsp}
(to which all-pairs shortest paths also reduces).
We establish a reduction from convolution to matrix multiplication.

\begin{theorem}\label{reduction}
  If we can compute the $(\min,-)$ matrix multiplication of two
  $n \times n$ matrices in $T(n)$ time, then we can compute the
  $(\min,-)$ convolution of two vectors of length $n$
  in $O((n + T(\sqrt n)) \sqrt n)$ time.
\end{theorem}

\begin{proof}
  We claim that computing the $(\min,-)$ convolution
  $\vec z = \vec x \minconvolve \vec y$
  reduces to the following $(\min,-)$ matrix multiplication:
  $$
  \def\matrixA{
    \begin{matrix}
      x_0 & x_1 & \cdots & x_{\sqrt n-1} \\
      x_{\sqrt n} & x_{\sqrt n+1} & \cdots & x_{2 \sqrt n-1} \\
      \vdots & \vdots & \ddots & \vdots \\
      x_{n-\sqrt n} & x_{n-\sqrt n+1} & \cdots & x_{n-1}
    \end{matrix}
  }
  \def\matrixB{
    \begin{matrix}
      y_{\sqrt n-1} & y_{\sqrt n} & \cdots & y_{n-2} & y_{n-1} \\
      y_{\sqrt n-2} & y_{\sqrt n-1} & \cdots & y_{n-3} & y_{n-2} \\
      \vdots & \vdots & \ddots & \vdots & \vdots \\
      y_1 & y_2 & \cdots & y_{n-\sqrt n-2} & y_{n-\sqrt n-1} \\
      y_0 & y_1 & \cdots & y_{n-\sqrt n-1} & y_{n-\sqrt n}
    \end{matrix}
  }
  P = 
    {\scriptstyle \sqrt n}
    \lefteqn{\phantom{\Bigg\{ \Bigg(}
             \underbrace{ \phantom{\matrixA} }_{\sqrt n}}
    \left\{ \left( \matrixA \right) \right.
  \minmultiply
    \lefteqn{\phantom{\Bigg(} \underbrace{ \phantom{\matrixB} }_{n-\sqrt n+1}}
    \left. \left( \matrixB \right) \right\} {\scriptstyle \sqrt n}.
  $$
  The $(i,j)$th entry $p_{i,j}$ of this product $P$ is
  $$
  p_{i,j} =
    \min_{m=0}^{\sqrt n-1} \left( x_{i\sqrt n + m} - y_{j+\sqrt n-1-m} \right).
  $$
  Let $\bar k = \lfloor k/\sqrt n\rfloor \sqrt n$
  denote the next smaller multiple of $\sqrt n$ from~$k$.
  Now, given the product $P$ above, we can compute each element $z_k$
  of the convolution $\vec z$ as follows:
  $$
  z_k = \min \left\{ \begin{array}{c}
          p_{0,k+1-\sqrt n},
          p_{1,k+1-2\sqrt n},
          p_{2,k+1-3\sqrt n}, \dots,
          p_{\lfloor k/\sqrt n\rfloor-1, k-\lfloor k/\sqrt n\rfloor \sqrt n}, \\
          x_{\bar k} - y_{k-\bar k},
          x_{\bar k+1} - y_{k-\bar k-1}, \dots,
          x_k - y_0
        \end{array} \right\}.
  $$
  This $\min$ has $O(\sqrt n)$ 
  terms, and thus $z_k$ can be computed in
  $O(\sqrt n)$ time.  The entire vector $\vec z$ can therefore
  be computed in $O(n \sqrt n)$ time, given the matrix product~$P$.

  It remains to show how to compute the rectangular product $P$ efficiently,
  given an efficient square-matrix $(\min,-)$ multiplication algorithm.
  We simply break the product $P$ into at most $\sqrt n$ products of
  $\sqrt n \times \sqrt n$ matrices: the left term is the entire left matrix,
  and the right term is a block submatrix.  The number of blocks is
  $\lceil (n-\sqrt n+1)/\sqrt n \rceil \leq \sqrt n$.
  Thus the running time for the product is $O(T(\sqrt n) \sqrt n)$.

  Summing the reduction cost and the product cost, we obtain a total cost of
  $O((n+T(\sqrt n)) \sqrt n)$.
\end{proof}

Plugging in $T(n) = O(n^3/\lg n)$ from \cite{Chan-2005-apsp} 
allows us to obtain an alternative proof of Theorem~\ref{min convolution RAM}.
Plugging in $T(n) = O(n^3 (\lg \lg n)^3/\lg^2 n)$ from \cite{Chan-2010-apsp}
immediately gives us the following improved result:

\begin{corollary} \label{min convolution word RAM} \sloppy
  The $(\min,-)$ convolution of two vectors of length $n$
  can be computed in $O(n^2 (\lg \lg n)^3/\lg^2 n)$ time on a real RAM.
\end{corollary}

Combining Theorem~\ref{l_infty reduction} and
Corollary~\ref{min convolution word RAM}, we obtain the following result:

\begin{corollary} \label{l_infty word RAM}
  The $\ell_\infty$ necklace alignment problem can be solved
  in\\ $O(n^2 (\lg \lg n)^3/\lg^2 n)$ time on a real RAM.
\end{corollary}

We remark that by the reduction in Theorem~\ref{reduction}, any nontrivial
lower bound for $(\min,-)$ convolution would imply a
lower bound for $(\min,-)$ matrix multiplication and 
the all-pairs shortest path problem.

\section{$\ell_1$ Necklace Alignment and $(\median,+)$ Convolution}
\label{sec:l_1}

\subsection{Reducing $\ell_1$ Necklace Alignment to $(\median,+)$ Convolution}

First we show the relation between $\ell_1$ necklace alignment
and $(\median,+)$ convolution.  We need the following basic fact~\cite{ardila-2008}:

\begin{fact} \label{median fact}
  For any vector $\vec z = \langle z_0, z_1, \dots, z_{n-1} \rangle$,
  $\displaystyle \sum_{i=0}^{n-1} |z_i + c|$ is minimized when
  $c = -\median_{i=0}^{n-1} z_i$.
\end{fact}

Instead of using $(\median,+)$ convolution directly,
we use the equivalent form, $(\median,-)$ convolution:

\begin{theorem} \label{l_1 reduction}
  The $\ell_1$ necklace alignment problem can be reduced in $O(n)$ time
  to one $(\median,-)$ convolution.
\end{theorem}

\begin{proof}
  For two necklaces $\vec x$ and $\vec y$, we apply the $(\median,-)$
  convolution
  to the following vectors,
  as in the proof of Theorem~\ref{l_infty reduction}:
  \begin{eqnarray*}
    \vec x' &=& \langle x_0, x_0, x_1, x_1, \dots, x_{n-1}, x_{n-1};
                        \underbrace{\infty, -\infty, \infty, -\infty, \dots,
                                    \infty, -\infty}_{2 n} \rangle, \\
    \vec y' &=& \langle y_{n-1}, y_{n-1}, y_{n-2}, y_{n-2}, \dots, y_0, y_0;
                        y_{n-1}, y_{n-1}, y_{n-2}, y_{n-2}, \dots, y_0, y_0
                \rangle.
  \end{eqnarray*}
  $\vec x' \medianconvolve \vec y'$ is
  $$ \median_{i=0}^{2(n+s')+1} (x'_i - y'_{2(n+s')+1-i})
   = \median_{i=0}^{n-1} (x_i - y_{(i-s'-1) \bmod n}),
  $$
  which is $\median_{i=0}^{n-1} (x_i - y_{(i+s) \bmod n})$
  if we let $s' = n-1-s$.
  Applying Fact~\ref{median fact}, we can therefore minimize
  $\median_{i=0}^{n-1} |x_i - y_{(i+s) \bmod n} + c|$ over~$c$,
  for each $s \in \{0, 1, \dots, n-1\}$.
  By brute force, we can minimize over $s$ as well using $O(n)$
  additional comparisons and time.
\end{proof}

Our results for $(\median,-)$ convolution use the following result
of Frederickson and Johnson:

\begin{theorem} {\rm \cite{Frederickson-Johnson-1982}} \label{median}
  The median element of the union of $k$ sorted lists, each of length~$n$,
  can be computed in $O(k \lg n)$ time and comparisons.
\end{theorem}

\subsection{$(\median,-)$ Convolution in Nonuniform Linear Decision Tree}

We begin with our results for the nonuniform linear decision tree model:

\begin{theorem} \label{median convolution nonuniform}
  The $(\median,-)$ convolution of two vectors of length $n$ can be computed
  in $O(n \sqrt {n \lg n})$ time in the nonuniform linear decision tree model.
\end{theorem}

\begin{proof}
  As in the proof of Theorem \ref{l_infty nonuniform},
  we sort the set $D = \{x_i - x_j, y_i - y_j : |i-j| \leq d\}$
  of pairwise differences between nearby $x_i$'s and nearby $y_i$'s,
  where $d \leq n$ is a value to be determined later.
  By Theorem \ref{Fredman lemma}, this step requires
  $O(n d + n \lg n)$ comparisons between differences.
  These comparisons enable us to compare $x_i - y_{k-i}$
  versus $x_j - y_{k-j}$ for free, provided $|i - j| \leq d$, because
  $x_i - y_{k-i} < x_j - y_{k-j}$ precisely if
  $x_i - x_j < y_{k-i} - y_{k-j}$.
  In particular, we can sort each list
  $$L_k(\lambda) = \Big\langle x_i - y_{k-i} ~\Big|~
                     i = \lambda, \lambda+1, \dots, \min\{\lambda+d,n\}-1
                   \Big\rangle$$
  for free.
  By Theorem~\ref{median}, we can compute the median of
  $L_k(0) \cup L_k(d) \cup L_k(2 d) \cup \cdots \cup L_k(\lceil k/d \rceil d)$,
  i.e., $\median_{i=0}^k (x_i - y_{k-i})$,
  in $O((k/d) \lg d) = O((n/d) \lg d)$ comparisons.
  Also, in the same asymptotic number of comparisons, we can binary search to
  find where the median fits in each of the $L_k(\lambda)$ lists, and therefore
  which differences are smaller and which differences are larger than
  the median.
  This median is the $k$th entry of $\vec x \medianconvolve \vec y$.
  Therefore, we can compute all $n$ entries of $\vec x \medianconvolve \vec y$
  in $O(n d + n \lg n + (n^2/d) \lg d)$ comparisons.
  This asymptotic running time is minimized
  when $n d = \Theta((n^2/d) \lg d)$,
  i.e., when $d^2/\lg d = \Theta(n)$.
  Substituting $d = \sqrt{n \lg n}$, we obtain a running time of
  $O(n \sqrt {n \lg n})$ in the nonuniform linear decision tree model.
\end{proof}

Combining Theorems~\ref{l_1 reduction} and \ref{median convolution nonuniform},
we obtain the following result:

\begin{corollary} \label{l_1 nonuniform}
  The $\ell_1$ necklace alignment problem can be solved in
  $O(n \sqrt {n \lg n})$ time in the nonuniform linear decision tree model.
\end{corollary}

\subsection{$(\min,-)$ Convolution in Real RAM via Geometric Dominance}

Now we turn to the analogous results for the real RAM:

\begin{theorem} \label{median convolution RAM}
  The $(\median,-)$ convolution of two vectors of length $n$ can be computed
  in $O(n^2 (\lg \lg n)^2/\lg n)$ time on a real RAM.
\end{theorem}

\begin{proof}
  Let $\vec x$ and $\vec y$ denote the two vectors of length~$n$,
  and let $\vec x \medianconvolve \vec y$ denote their $(\median,-)$
  convolution.
  For each permutation $\pi$ on the set $\{0, 1, \dots, d-1\}$,
  for each $i \in \{0, d, 2 d, \dots, \lfloor n/d \rfloor d\}$,
  and for each $j \in \{0, 1, \dots, n-1\}$,
  we define the $(d-1)$-dimensional points
  $$\arraycolsep=0.5\arraycolsep
  \begin{array}{rcllll}
    p_{\pi,i} &=& (x_{i+\pi(0)} - x_{i+\pi(1)}, & x_{i+\pi(1)} - x_{i+\pi(2)},
       & \dots, &  x_{i+\pi(d-2)} - x_{i+\pi(d-1)}), \\
    q_{\pi,j} &=& (y_{j-\pi(0)} - y_{j-\pi(1)}, & y_{j-\pi(1)} - y_{j-\pi(2)},
       & \dots, &  y_{j-\pi(d-2)} - y_{j-\pi(d-1)}),
  \end{array}$$
  (To handle boundary cases, define $x_i=\infty$ and $y_j=-\infty$
  for indices $i,j$ outside $[0,n-1]$.)
  For each permutation $\pi$, we apply Lemma~\ref{chan} to the set of red
  points $\{p_{\pi,i} : i = 0, d, 2 d, \dots, \lfloor n/d \rfloor d\}$ and
  the set of blue points $\{q_{\pi,j} : j = 0, 1, \dots, n-1\}$,
  to obtain all dominating pairs $(p_{\pi,i},q_{\pi,j})$.

  Point $p_{\pi,i}$ dominates $q_{\pi,j}$ precisely if
  $x_{i+\pi(\delta)} - x_{i+\pi(\delta+1)} \geq
   y_{j-\pi(\delta)} - y_{j-\pi(\delta+1)}$
  for all $\delta \in \{0, 1, \dots, d-2\}$
  (ignoring the indices outside $[0, n-1]$).
  By re-arranging terms, this condition is equivalent to
  $x_{i+\pi(\delta)} - y_{j-\pi(\delta)} \geq
   x_{i+\pi(\delta+1)} - y_{j-\pi(\delta+1)}$
  for all $\delta \in \{0, 1, \dots, d-2\}$,
  i.e., $\pi$ is a sorting permutation of
  $\langle x_i - y_j, x_{i+1} - y_{j-1}, \dots,
           x_{i+d-1} - y_{j-d+1} \rangle$.
  If we substitute $j = k-i$, we obtain that $(p_{\pi,i}, q_{\pi,k-i})$ is a
  dominating pair precisely if $\pi$ is a sorting permutation of the list
  $L_k(i) = \langle x_i - y_{k-i}, x_{i+1} - y_{k-i+1}, \dots,
           x_{\min\{i+d,n\}-1} - y_{\min\{k-i+d,n\}-1} \rangle$.
  Thus, the set of dominating pairs gives us the sorted order of
  $L_k(i)$ for each $i$ divisible by $d$ and for each~$k$.
  Also, there can be at most $O(n^2/d)$ total dominating pairs
  $(p_{\pi,i}, q_{\pi,j})$ over all $i,j,\pi$,
  because there are $O(n/d)$ choices for $i$~and
  $O(n)$ choices for~$j$,
  and if $(p_{\pi,i},q_{\pi,j})$ is a dominating pair,
  then $(p_{\pi',i},q_{\pi',j})$ cannot be a dominating pair
  for any permutation $\pi' \neq \pi$.
  (Here we assume that the sorted order is unique, which can be arranged
   by standard perturbation techniques or by breaking ties consistently
   \cite{Chan-2005-apsp}.)
  Hence, the running time of the $d!$ executions of Lemma~\ref{chan} is
  $d! \, 2^{O(d)} n^{1+\epsilon} + O(n^2/d)$ time,
  which is $O(n^2 \lg \lg n/\lg n)$
  if we choose $d = \alpha \lg n / \lg \lg n$
  for a sufficiently small constant $\alpha > 0$.
  By Theorem~\ref{median}, we can compute the median of
  $L_k(0) \cup L_k(d) \cup L_k(2 d) \cup \cdots \cup L_k(\lceil k/d \rceil d)$,
  i.e., $\median_{i=0}^k (x_i - y_{k-i})$,
  in $O((k/d) \lg d) = O((n/d) \lg d)$ comparisons.
  Also, in the same asymptotic number of comparisons, we can binary search to
  find where the median fits in each of the $L_k(\lambda)$ lists, and therefore
  which differences are smaller and which differences are larger than
  the median.
  This median is the $k$th entry of $\vec x \medianconvolve \vec y$.
  Therefore all $n$ entries can be computed in
  $O(n^2 (\lg d)/d) = O(n^2 (\lg \lg n)^2/\lg n)$ time
  on a real RAM.
\end{proof}

Combining Theorems~\ref{l_1 reduction} and \ref{median convolution RAM},
we obtain the following result:

\begin{corollary} \label{l_1 RAM}
  The $\ell_1$ necklace alignment problem can be solved in\\
  $O(n^2 (\lg \lg n)^2/\lg n)$ time on a real RAM.
\end{corollary}

As before, this approach likely cannot be improved beyond $O(n^2 / \lg n)$,
because such an improvement would require an improvement to Lemma~\ref{chan},
which would in turn improve the fastest known algorithm for all-pairs shortest
paths in dense graphs \cite{Chan-2010-apsp}.

In contrast to $(\median,+)$ convolution, $(\mean,+)$ convolution
is trivial to compute in linear time by inverting the two summations.

\section{Conclusion}
\label{Conclusion}

The convolution problems we consider here have connections to many classic
problems, and it would be interesting to explore whether the structural
information extracted by our algorithms could be used to devise faster
algorithms for these classic problems.
For example, does the antidiagonal information of the $X+Y$ matrix
lead to a $o(n^2 \lg n)$-time algorithm for sorting $X+Y$?
We believe that any further improvements to our convolution algorithms
would require progress and/or have interesting implications on
all-pairs shortest paths \cite{Chan-2005-apsp}.

Our $(\min,-)$-convolution algorithms give subquadratic algorithms
for \emph{polyhedral 3SUM}:
given three lists,
$A = \langle a_0, a_1, \dots, a_{n-1} \rangle$,
$B = \langle b_0, b_1, \dots, b_{n-1} \rangle$,
and
$C = \langle c_0, c_1, \dots, c_{2n-2} \rangle$,
such that $a_i + b_j \leq c_{i+j}$ for all $0 \leq i, j < n$,
decide whether $a_i + b_j = c_{i+j}$ for any $0 \leq i, j < n$.
This problem is a special case of 3SUM,
and this special case has an $\Omega(n^2)$ lower bound in the
3-linear decision tree model \cite{Erickson-1999-satisfiability}.
Our results solve polyhedral 3SUM in $O(n^2 / \lg n)$ time
in the 4-linear decision tree model, and in $O(n \sqrt n)$ time
in the nonuniform 4-linear decision tree model,
solving an open problem of Erickson \cite{Demaine-O'Rourke-2005-open}.
Can these algorithms be extended to solve 3SUM in subquadratic time
in the (nonuniform) decision tree model?

\section*{Acknowledgments}

This work was initiated at the 20th Bellairs Winter Workshop on
Computational Geometry held January 28--February 4, 2005.
We thank the other participants of that workshop---Greg Aloupis,
Justin Colannino, Mirela Damian-Iordache, Vida Dujmovi\'c,
Francisco Gomez-Martin, Danny Krizanc, Erin \allowbreak McLeish, Henk Meijer, Patrick Morin,
Mark Overmars, Suneeta Ramaswami, David Rappaport, Diane Souvaine,
Ileana Streinu, David Wood, Godfried Toussaint, Remco Veltkamp, and
Sue Whitesides---for helpful discussions
and contributing to a fun and creative atmosphere.
We particularly thank the organizer, Godfried Toussaint,
for posing the problem to us.
The last author would also like to thank Luc Devroye for pointing out 
the easy generalization of the $\ell_2$ necklace alignment problem 
to $\ell_p$ for any fixed even integer $p$.

\bibliography{3sum,algs,compgeom,convolution,juggling,math,matrix,rhythm,shortestpaths,strings,sorting}
\bibliographystyle{spmpsci}

\end{document}